\documentclass[sigconf,nonacm]{acmart}

\newcommand{\ourTitle}{Solving k-center Clustering (with Outliers) in MapReduce and Streaming, almost as Accurately as Sequentially.}

\usepackage[utf8]{inputenc}

\usepackage{enumitem}
\usepackage{xspace}
\usepackage[ruled,vlined]{algorithm2e}
\usepackage{booktabs}

\usepackage{hyperref}
\hypersetup{
  colorlinks=true,
  linkcolor=black,
  citecolor=black,
  urlcolor=black
}

\usepackage{tikz}
\usepackage{pgfplots}
\usetikzlibrary{shapes}
\usetikzlibrary{calc}
\usepgfplotslibrary{fillbetween}
\usetikzlibrary{arrows.meta}
\usetikzlibrary{fadings}
\tikzfading[name=fade in,
inner color=transparent!100,
outer color=transparent!0]

\newcommand{\opt}{r_k^*(S)}
\newcommand{\opto}{r_{k, z}^*(S)}
\newcommand{\gmm}{\textsc{gmm}\xspace}
\newcommand{\outliers}{{\sc OutliersCluster}\xspace}
\newcommand{\BO}[1]{O\left(#1\right)}

\newcommand{\BT}[1]{\Theta\left(#1\right)}
\renewcommand{\epsilon}{\varepsilon}

\DeclareMathOperator*{\argmax}{arg\,max}

\SetKwComment{Comment}{$\triangleright$\ }{}
\newcommand{\Let}[2]{#1 $\leftarrow$ #2}
\SetKw{To}{to}
\SetKw{Error}{Error}
\SetKwProg{Map}{Map:}{ }{}
\SetKwProg{Reduce}{Reduce:}{ }{}
\SetKwProg{Procedure}{Procedure}{ }{end}
\SetKw{Emit}{emit}

\newtheorem{theorem}{Theorem}
\newtheorem{lemma}{Lemma}

\newtheorem{corollary}{Corollary}

\newcommand{\dataset}[1]{\texttt{#1}}

\newif\ifextended
\extendedtrue

\begin{document}

\title{\ourTitle}

\author{Matteo Ceccarello}
\affiliation{
  \institution{IT University and BARC}
  \city{Copenhagen}
  \state{Denmark}
}
\email{mcec@itu.dk}

\author{Andrea Pietracaprina}
\affiliation{
  \institution{University of Padova}
  \city{Padova}
  \state{Italy}
}
\email{andrea.pietracaprina@unipd.it}

\author{Geppino Pucci}
\affiliation{
  \institution{University of Padova}
  \city{Padova}
  \state{Italy}
}
\email{geppo@dei.unipd.it}

\begin{abstract}
Center-based clustering is a fundamental primitive for data analysis
and becomes very challenging for large datasets. In this paper, we
focus on the popular $k$-center variant which, given a set $S$ of points
from some metric space and a parameter $k<|S|$, requires to identify a
subset of $k$ centers in $S$ minimizing the maximum distance of any point
of $S$ from its closest center. A more general formulation, introduced
to deal with noisy datasets, features a further parameter $z$ and allows
up to $z$ points of $S$ (outliers) to be disregarded when computing the
maximum distance from the centers. We present coreset-based 2-round
MapReduce algorithms for the above two formulations of the problem,
and a 1-pass Streaming algorithm for the case with outliers. For any
fixed $\epsilon>0$, the algorithms yield solutions whose approximation
ratios are a mere additive term $\epsilon$ away from those achievable by
the best known polynomial-time sequential algorithms, a result that
substantially improves upon the state of the art. 
Our algorithms are rather simple and adapt to the intrinsic complexity
of the dataset, captured by the doubling dimension $D$ of the metric
space. Specifically, our analysis shows that the algorithms become
very space-efficient for the important case of small (constant)
$D$. These theoretical results are complemented with a set of
experiments on real-world and synthetic datasets of up to over a
billion points, which show that our algorithms yield better quality
solutions over the state of the art while featuring excellent
scalability, and that they also lend themselves to sequential
implementations much faster than existing ones.
\end{abstract}

\maketitle

\section{Introduction} \label{sec-introduction}
\noindent \sloppy Center-based clustering is a fundamental
unsupervised learning primitive for data management, with applications
in a variety of domains such as database search, bioinformatics,
pattern recognition, networking, facility location, and many more
\cite{HennigMMR15}.  Its general goal is to partition a set of data
items into groups according to a notion of similarity, captured by
closeness to suitably chosen group representatives, called
centers. There is an ample and well-established literature on
sequential strategies for different instantiations of center-based
clustering \cite{AwasthiB15}. However, the explosive growth of data
that needs to be processed often rules out the use of these strategies
which are efficient on small-sized datasets, but impractical on large
ones. Therefore, it is of paramount importance to devise efficient
clustering strategies tailored to the typical computational frameworks
for big data processing, such as MapReduce and Streaming
\cite{LeskovecRU14}.

In this paper, we focus on the \emph{$k$-center} problem, formally
defined as follows.  Given a set $S$ of points in a metric space and a
positive integer $k < |S|$, find a subset $T \subseteq S$ of $k$
points, called \emph{centers}, so that the maximum distance between
any point of $S$ to its closest center in $T$ is minimized. (Note that the association of each point to the closest center naturally
defines a clustering of $S$.)  Along with $k$-median and $k$-means,
which require to minimize, respectively, the sum of all distances and
all square distances to the closest centers, $k$-center is
a very popular instantiation of center-based clustering which has
recently proved a pivotal primitive for data and graph analytics
\cite{IndykMMM14,AghamolaeiFZ15,CeccarelloPPU15,CeccarelloPPU16,CeccarelloPPU17,CeccarelloFPPV17},
and whose efficient solution in the realm of big data has attracted a
lot of attention in the literature
\cite{Charikar2001,McCutchen2008,EneIM11,MalkomesKCWM15}.

The $k$-center problem is NP-hard \cite{Gonzalez85},
therefore one has to settle for approximate solutions. Also, since its
objective function involves a maximum, the solution is at risk of
being severely influenced by a few ``distant'' points, called
\emph{outliers}.  In fact, the presence of outliers is inherent in
many datasets, since these points are often artifacts of data
collection, or represent noisy measurements, or simply erroneous
information.  To cope with this problem, $k$-center admits
a formulation that takes into account outliers \cite{Charikar2001}:
when computing the objective function, up to $z$ points are allowed to
be discarded, where $z$ is a user-defined input parameter.

A natural approach to compute approximate solutions to large instances
of combinatorial optimization problems entails efficiently extracting
a much smaller subset of the input, dubbed \emph{coreset}, which
contains a good approximation to the global optimum, and then applying
a standard sequential approximation algorithm to such a coreset.  The
benefits of this approach are evident when the coreset construction is
substantially more efficient than running the (possibly very
expensive) sequential approximation algorithm directly on the whole
input, so that significant performance improvements are attained by
confining the execution of such algorithm on a small subset of the data.  Using coresets
much smaller than the input, the authors of
\cite{MalkomesKCWM15} present MapReduce algorithms for the $k$-center
problem with and without outliers, whose (constant) approximation
factors are, however, substantially larger than their best sequential
counterparts. In this work, we further leverage the coreset approach
and unveil interesting tradeoffs between the coreset size and the
approximation quality, showing that better approximation is achievable
through larger coresets.  The obtainable tradeoffs are regulated by
the doubling dimension of the underlying metric space and allow us to
obtain improved MapReduce and Streaming algorithms for the two
formulations of the $k$-center problem, whose approximation ratios can be
made arbitrarily close to the one featured by the best sequential
algorithms. Also, as a by-product, we obtain a sequential
algorithm for the case with outliers which is 
considerably faster than existing ones.

\subsection{Related work} \label{subsec:relwork}

Back in the 80's, Gonzalez \cite{Gonzalez85} developed a very
popular 2-approximation sequential algorithm for the $k$-center
problem  running in  $\BO{k|S|}$ time, which is referred to as \gmm in
the recent literature. In the same paper, the author showed that it is
impossible to achieve an approximation factor $2-\epsilon$, for fixed
$\epsilon >0$, in general metric spaces, unless $P=NP$.  To deal with
noise in the dataset, Charikar et al.~\cite{Charikar2001} introduced
the $k$-center problem with $z$ outliers, where the clustering is
allowed to ignore $z$ points of the input.  For this problem, they
gave a 3-approximation algorithm which runs in $\BO{k|S|^2 \log |S|}$ time.
Furthermore, they proved that, for this problem, it is impossible to
achieve an approximation factor $3-\epsilon$, for fixed $\epsilon >0$,
in general metric spaces, unless $P=NP$.

With the advent of big data, a lot of attention has been devoted to
the MapReduce model of computation, where a set of processors with
limited-size local memories process data in a sequence of parallel
rounds~\cite{DeanG08,PietracaprinaPRSU12,LeskovecRU14}.  The $k$-center problem
under this model was first studied by Ene et al.~\cite{EneIM11}, who
provided a 10-approximation randomized algorithm.  This result was
subsequently improved in~\cite{MalkomesKCWM15} with a deterministic
4-approximation algorithm requiring an $\BO{\sqrt{|S|k}}$-size local
memory.  As for the $k$-center problem with $z$ outliers, a
deterministic $13$-approximation MapReduce algorithm was 
presented in \cite{MalkomesKCWM15}, requiring an
$\BO{\sqrt{|S|(k+z)}}$-size local memory.  We remark that randomized
multi-round MapReduce algorithms for the two formulations of the $k$-center
problem, with approximation ratios $2$ and 4 respectively, have been
claimed but not described in the short communication \cite{ImM15}.
While, theoretically, the MapReduce algorithms proposed in our work  seem competitive with
respect to both round complexity and space requirements with the algorithms announced in \cite{ImM15}, 
any comparison is clearly subject to the availability of more details.

As mentioned before, the algorithms in \cite{MalkomesKCWM15} are based on the use of
(\emph{composable}) \emph{coresets}, a very useful tool in the
MapReduce setting \cite{Agarwal2004,IndykMMM14}.  For a given
objective function, a coreset is a small subset extracted from the input
which embodies a solution whose cost is close to the cost of the
optimal solution on the whole set.  The additional property of
composability requires that, if coresets are extracted from distinct
subsets of a given partition of the input, their union embodies a
close-to-optimal solution of the whole input.  Composable coresets
enable the development of parallel algorithms, where each processor
computes the coreset relative to one subset of the partition, and
the computation of the final solution is then performed by one
processor that receives the union of the coresets.  Composable
coresets have been used for a number of problems, including diversity
maximization~\cite{IndykMMM14,AghamolaeiFZ15,CeccarelloPPU17,CeccarelloPP18},
submodular maximization~\cite{Mirrokni2015}, graph matching and
vertex cover~\cite{Assadi2017}. In  \cite{BadoiuHI02} the authors
provide a coreset-based $(1+\epsilon)$-approximation sequential algorithm to the
$k$-center problem for $d$-dimensional Euclidean spaces, whose time is
exponential in $k$ and $(1/\epsilon)^2$ and linear in $d$ and $|S|$.
However, the coreset construction is rather involved, not easily parallelizable and the 
resulting algorithm seems to be mainly of theoretical interest. 

\sloppy
Another option when dealing with large amounts of data is to process
the data in a streaming fashion. In the Streaming model, algorithms use
a single processor with limited working memory and are allowed only a
few sequential passes over the input (ideally just
one)~\cite{HenzingerRR98,LeskovecRU14}.  Originally developed for the external
memory setting, this model also captures the scenario in which data is
generated on the fly and must be analyzed in real-time, for instance
in a streamed DMBS or in a social media platform (e.g., Twitter trends
detection).  Under this model, Charikar et al.~\cite{CharikarCFM04}
developed a 1-pass algorithm for the $k$-center problem which
requires $\BT{k}$ working memory and computes an 8-approximation,
deterministically, or a 5.43-approximation, probabilistically. Later,
the result was improved in~\cite{McCutchen2008} attaining a
$(2+\epsilon)$ approximation, deterministically, needing a working
memory of size $\BT{k\epsilon^{-1}\log(\epsilon^{-1})}$.  In the same paper, the
authors give a deterministic $(4+\epsilon)$-approximation Streaming
algorithm for the formulation with $z$ outliers, which requires
$\BO{kz\epsilon^{-1}}$ working memory.

\subsection{Our contribution} \label{sec:contribution}

The coreset-based MapReduce algorithms of \cite{MalkomesKCWM15} for
$k$-center, with and without outliers, use the \gmm\ sequential
approximation algorithm for $k$-center in a ``bootstrapping'' fashion:
namely, in a first phase, a set of $k$ centers ($k+z$ centers in the
case with $z$ outliers) is determined in each subset of an arbitrary
partition of the input dataset, and then the final solution is
computed on the coreset provided by the union of these centers, using
a sequential approximation algorithm for the specific problem
formulation. Our work is motivated by the following natural question:
what if we select more centers from each subset of the partition in
the first phase?  Intuitively, we should get a better solution than if
we just selected $k$ (resp., $k+z$) centers. In fact, selecting more
and more centers from each subset should yield a solution
progressively closer to the one returned by the best sequential
algorithm on the whole input, at the expense of larger space
requirements. 

This paper provides a thorough characterization of the space-accuracy
tradeoffs achievable by exploiting the aforementioned idea for both
formulations of the $k$-center problem (with and without outliers).
We present improved MapReduce and Streaming algorithms
which leverage a judicious selection of larger (composable) coresets
to boost the quality of the solution embodied in the (union of the)
coresets. We analyze the memory requirements of our algorithms
in terms of the desired
approximation quality, captured by a precision parameter $\epsilon$,
and of the \emph{doubling dimension} $D$ of the underlying metric
space, a parameter that generalizes the dimensionality of Euclidean
spaces to arbitrary metric spaces and is thus related to the 
difficulty of spotting good clusterings.
We remark that this kind of parametrized analysis is particularly
relevant in the realm of big data, where distortions introduced to
account for worst-case scenarios may be too extreme to provide
meaningful insights on actual algorithm's performance, and
it has been employed in a variety of contexts
including diversity maximization, 
clustering, nearest neighbour search, routing, machine learning, 
and graph analytics (see \cite{CeccarelloPPU17} and references therein).

Our specific results are the following:
\begin{itemize} \sloppy
\item
A deterministic 2-round, $(2+\epsilon)$-approximation MapReduce
algorithm for the $k$-center problem, which requires
$\BO{\sqrt{|S|k}(4/\epsilon)^D}$ local memory.
\item 
A deterministic 2-round, $(3+\epsilon)$-approximation  MapReduce algorithm
for the $k$-center problem with $z$ outliers, 
which requires
$\BO{\sqrt{|S|(k+z)}(24/\epsilon)^D}$ local memory.
\item 
A randomized 2-round, $(3+\epsilon)$-approximation  MapReduce algorithm
for the $k$-center problem with $z$ outliers, which
reduces the local memory requirements to 
$\BO{\left(\sqrt{|S|(k+\log|S|)}+z\right)(24/\epsilon)^D}$.
\item
A deterministic 1-pass, $(3+\epsilon)$-approximation Streaming algorithm for
the $k$-center problem with $z$ outliers, which requires
$\BO{(k+z)(96/\epsilon)^D}$  working memory.
\end{itemize}
Using our coreset constructions we can also attain 
a $(2+\epsilon)$-approximation Streaming algorithm 
for $k$-center without outliers, which however
would not improve on the state-of-the-art algorithm \cite{McCutchen2008}.
Nonetheless, for the sake of completeness, we will compare these two algorithms
experimentally in Section~\ref{sec:experiments}.

Observe that for both formulations of the problem, our algorithms
feature approximation guarantees which are a mere additive
term $\epsilon$  larger than the best achievable sequential guarantee, and yield
substantial quality improvements over the state-of-the-art
\cite{MalkomesKCWM15,McCutchen2008}. Moreover, the randomized
MapReduce algorithm for the formulation with outliers
features smaller coresets, thus attaining a reduction in the
local memory requirements which becomes
substantial in plausible scenarios where the number
of outliers $z$ (e.g., due to noise) is considerably larger than the
target number $k$ of clusters, although much smaller than the input
size.  

While our algorithms are applicable to general metric spaces, on
spaces of constant doubling dimension $D$ and for constant $\epsilon$,
their local space/working memory requirements are polynomially
sublinear in the dataset size, in the MapReduce setting, and
independent of the dataset size, in the Streaming setting.  Moreover,
a very desirable feature of our MapReduce algorithms is that they are
\emph{oblivious to $D$}, in the sense that the value $D$ (which may be
not known in advance and hard to evaluate) is not used explicitly in
the algorithms but only in their analysis. In contrast, the 1-pass
Streaming algorithm makes explicit use of $D$, although we will show
that it can be made oblivious to $D$ at the expense of one extra pass
on the input stream.

As a further important result, the MapReduce algorithm
for the case with outliers admits a direct sequential implementation
which substantially improves the time performance of the state-of-the-art
algorithm by \cite{Charikar2001}
while essentially preserving the approximation quality.

We also provide experimental evidence of the competitiveness of our
algorithms on real-world and synthetic datasets of up to over a
billion points, comparing with baselines set by the algorithms
in\ \cite{MalkomesKCWM15} for MapReduce, and\ \cite{McCutchen2008} for
Streaming.  In the MapReduce setting, the experiments show that
tighter approximations over the algorithms in \cite{MalkomesKCWM15}
are indeed achievable with larger coresets.  In fact, while our
theoretical bounds on the space requirements embody large constant factors,
the improvements in the approximation quality are already noticeable
with a modest increase of the coreset size.  In the Streaming setting,
for $k$-center without outliers we show that the
$(2+\epsilon)$-approximation algorithm based on our techniques is
comparable to\ \cite{McCutchen2008}, whereas for $k$-center with
outliers we obtain solutions of better quality using significantly
less memory and time.  The experiments also show that the Streaming
algorithms feature high-throughput, and that the MapReduce algorithms
exhibit high scalability.  Finally, we show that, indeed, implementing
our coreset strategy sequentially yields a substantial running time
improvement with respect to the state-of-the art algorithm
\cite{Charikar2001}, while preserving the approximation quality.
\\[0.2cm] {\bf Organization of the paper} The
rest of the paper is organized as follows.  
Section~\ref{sec:prelim} contains a number of preliminary concepts.
Section~\ref{sec:MR} and Section~\ref{sec:streaming} present,
respectively, our MapReduce and Streaming algorithms.  The
experimental results are reported in
Section~\ref{sec:experiments}. Finally, Section~\ref{sec:conclusions}
offers some concluding remarks.

\section{Preliminaries} \label{sec:prelim}

Consider a metric space $\mathcal{S}$ with distance function $d(\cdot,
\cdot)$. For a point $u \in \mathcal{S}$, the \emph{ball of radius $r$
  centered at $u$} is the set of points at distance at most $r$ from
$u$.  The \emph{doubling dimension} of $\mathcal{S}$ is the smallest
$D$ such that for any radius $r$ and point $u \in \mathcal{S}$, all
points in the ball of radius $r$ centered at $u$ are included in the
union of at most $2^D$ balls of radius $r/2$ centered at suitable
points.  It immediately follows that, for any $0 < \epsilon \le 1$, a
ball of radius $r$ can be covered by at most $(1/\epsilon)^D$ balls of
radius $\epsilon r$. Notable examples of metric spaces with bounded
doubling dimension are Euclidean spaces and spaces induced by
shortest-path distances in mildly-expanding topologies. Also,
the notion of doubling dimension can be defined for
an individual dataset and 
it may turn out much lower than the one of the underlying 
metric space (e.g., a set of collinear points 
in $\Re^2$). In fact, the
space-accuracy tradeoffs of our algorithms only depend on the doubling
dimension of the input dataset.

Define the distance between a point $s \in \mathcal{S}$ and a set
$X \subseteq \mathcal{S}$ as $d(s, X) = \min_{x \in X}d(s, x)$.
Consider now a dataset $S\subseteq  \mathcal{S}$ and a subset  $T\subseteq S$.
We define the \emph{radius of $S$ with respect to $T$} as
\[
  r_T(S) = \max_{s \in S}d(s, T).
\]
The \emph{$k$-center problem} requires to find a subset
$T \subseteq S$ of size $k$ such that $r_{T}(S)$ is minimized. We
define $r_k^*(S)$ as the radius achieved by the optimal solution to the
problem. Note that $T$ induces immediately a partition of $S$
into $k$ clusters by assigning each point to its 
closest center, and we say that $r_T(S)$ is the radius
of such a clustering. 

In Section~\ref{subsec:relwork}  we mentioned the \gmm
algorithm~\cite{Gonzalez85}, which provides a sequential 2-approximation to the
$k$-center problem.  Here we briefly review how \gmm works.  Given a set
$S$, \gmm builds a set of centers $T$ incrementally in $k$
iterations.  An arbitrary point of $S$ is selected as the first center
and is added to $T$.  Then, the algorithm iteratively selects the next
center as the point with maximum distance from $T$, and adds it to $T$,
until $T$ contains $k$ centers.
Note that, rather than setting $k$ \emph{a priori}, \gmm can be used to
grow the set $T$ until a target radius is achieved.  In fact, the
radius of $S$ with respect to the set of centers $T$ incrementally built by \gmm
is a non-increasing function of the iteration number. In this paper,
we will make use of the following property of \gmm which bounds
its accuracy when run on a subset of the data.

\begin{lemma}\label{lem:gmm-subset}
Let $X \subseteq S$. For a given $k$, let $T_X$ be the output of \gmm when run on $X$. We have $r_{T_X}(X) \le 2 \cdot \opt$.
\end{lemma}
\begin{proof}
We prove this lemma by rephrasing the proof by Gonzalez~\cite{Gonzalez85} in
terms of subsets.  We need to prove that, $\forall x \in X$, $d(x,
T_X) \le 2 \cdot \opt$.  Assume by contradiction that this is not the
case.  Then, for some $y \in X$ it holds that $d(y,
T_X) > 2 \cdot \opt$.  By the greedy choice of \gmm, we have that for
any pair $t_1, t_2 \in T_X$, $d(t_1, t_2) \ge d(y, T_X)$, otherwise $y$
would have been included in $T_X$. So we have that $d(t_1, t_2) > 2
\cdot \opt$.  Therefore, the set $\{y\} \cup T_X$ consists  of $k+1$
points at distance $> 2\cdot\opt$ from each other.  Consider now the
optimal solution to $k$-center on the set $S$.  Since $(\{y\} \cup
T_X) \subseteq S$, two of the $k+1$ points of $\{y\} \cup T_X$, say $x_1$
and $x_2$, must be closest to the same optimal center $o^*$.  By the
triangle inequality we have $2\cdot\opt < d(x_1, x_2) \le d(x_1, o^*) + d(o^*, x_2) \le 2
\cdot \opt$, a contradiction.
\end{proof}

For a given set $S \subseteq \mathcal{S}$, 
the \emph{$k$-center problem with $z$ outliers} 
requires to identify a set $T$ of $k$ centers 
which minimizes
\[
  r_{T,Z_T}(S) = \max_{s \in S \setminus Z_T} d(s, T),
\]
where $Z_T$ is the set of $z$ points in $S$ with largest distance from
$T$ (ties broken arbitrarily).  In other words, the problem allows to
discard up the $z$ farthest points when computing the radius of the
set of centers, hence of its associated clustering.  For given $S$,
$k$, and $z$, we denote the radius of the optimal solution of this
problem by $r_{k, z}^*(S)$.  It is straightforward to argue that the
optimal solution of the problem without outliers with $k+z$ centers
has a smaller radius than the optimal solution of the problem with $k$
centers and $z$ outliers, that is
\begin{equation}
r_{k+z}^*(S) \le r_{k,z}^*(S).
\label{eq:radius-relation}
\end{equation}

\subsection{Computational frameworks}
A \emph{MapReduce}
algorithm~\cite{DeanG08,PietracaprinaPRSU12,LeskovecRU14} executes in
a sequence of parallel \emph{rounds}. In a round, a multiset $X$ of
key-value pairs is first transformed into a new multiset $X'$ of
key-value pairs by applying a given \emph{map function} (simply called
\emph{mapper}) to each individual pair, and then into a final multiset
$Y$ of pairs by applying a given \emph{reduce function} (simply called
\emph{reducer}) independently to each subset of pairs of $X'$ having
the same key.  The model features two parameters, $M_L$, the
\emph{local memory} available to each mapper/reducer, and $M_A$, the
\emph{aggregate memory} across all mappers/reducers. In our
algorithms, mappers are straightforward constant-space
transformations, thus the memory requirements will be related to the
reducers.  We remark that the MapReduce algorithms presented in this
paper also afford an immediate implementation and similar analysis in
the \emph{Massively Parallel Computation} (MPC) model
\cite{BeameKS13}, which is popular in the database community.

In the \emph{Streaming} framework \cite{HenzingerRR98,LeskovecRU14} the computation is
performed by a single processor with a small working memory, and
the input is provided as a continuous stream of items which is usually
too large to fit in the working memory. Multiple passes on the input
stream may be allowed.  Key performance indicators are the size of the
working memory and the number of passes.

The holy grail of big data algorithmics is the development of
MapReduce (resp., Streaming)
algorithms which work in as few rounds (resp., passes) as possible and 
require substantially sublinear local memory (resp.,
working memory) and linear aggregate memory.

\section{MapReduce algorithms} \label{sec:MR}

The following subsections present our MapReduce algorithms for the
$k$-center problem (Subsection~\ref{sec-noout}) and the $k$-center
problem with $z$ outliers (Subsection~\ref{sec-out}). The algorithms
are based on the use of composable coresets, which were reviewed in
the introduction, and can be viewed as improved variants of those by
\cite{MalkomesKCWM15}. The main novelty of our algorithms
is their leveraging a judiciously increased coreset size to attain
approximation qualities that are arbitrarily close to the ones
featured by the best known sequential algorithms. Also, in the
analysis, we relate the required coreset size to the doubling
dimension of the underlying metric space (whose explicit knowledge, however, is
not required by the algorithms) showing that coreset sizes stay
small for spaces of bounded doubling dimension.

\subsection{MapReduce algorithm for $k$-center} \label{sec-noout}
Consider an instance $S$ of the $k$-center problem and fix a precision
parameter $\epsilon \in (0, 1]$, which will be used to regulate the
  approximation ratio.  The MapReduce algorithm works in two rounds.
  In the first round, $S$ is partitioned into $\ell$ subsets $S_i$ of
  equal size, for $1 \leq i \leq \ell$. In parallel, on each $S_i$ we
  run \gmm incrementally and call $T_i^j$ the set of $j$ centers
  selected in the first $j$ iterations of the algorithm. Let
  $r_{T_i^k}(S_i)$ denote the radius of the set $S_i$ with respect to
  the first $k$ centers.  We continue to run \gmm until the first
  iteration $\tau_i \geq k$ such that $r_{T_i^{\tau_i}}(S_i) \le
  \epsilon/2 \cdot r_{T_i^k}(S_i)$, and define the coreset $T_i =
  T_i^{\tau_i}$.  In the second round, the union of the coresets $T =
  \bigcup_{i=1}^\ell T_i$ is gathered into a single reducer and \gmm
  is run on $T$ to compute the final set of $k$ centers.  In what
 follows, we show that these centers are a good solution to the
  $k$-center problem on $S$.

The analysis relies on the following two lemmas which state that each
input point has a close-by representative in $T$ and that $T$ has
small size.  We define a \emph{proxy function} $p: S \to T$ that maps
each $s \in S_i$ into the closest point in $T_i$, for every $1 \leq i
\leq \ell$.  The following lemma is an easy consequence of
Lemma~\ref{lem:gmm-subset}.
\begin{lemma}\label{lem:k-center-proxy}
For each $s \in S$, 
$d(s, p(s)) \le \epsilon \cdot \opt$.
\end{lemma}
\begin{proof}
Fix $i \in [1, \ell]$, and consider $S_i \subseteq S$, and the set
$T_i^k$ computed by the first $k$ iterations of \gmm.  Since $S_i$ is
a subset of $S$, by Lemma~\ref{lem:gmm-subset} we have that
$r_{T_i^k}(S_i) \le 2\cdot \opt$.  By construction, we have that
$r_{T_i}(S_i) \le \epsilon/2 \cdot r_{T_i^k}(S_i)$, hence
$r_{T_i}(S_i) \le \epsilon \opt$.  Consider now the proxy function
$p$.  For every $1 \leq i \leq \ell$ and $s \in S_i$, 
it holds that $d(s, p(s)) \le r_{T_i}(S_i) \le
\epsilon \opt$. 
\end{proof}

We can conveniently bound the size of $T$, the union of the coresets, as a function of the
doubling dimension of the underlying metric space. 
\begin{lemma}\label{lem:k-center-size}
If $S$ belongs to a metric space of doubling dimension
$D$, then
\[
|T| \le \ell\cdot k \cdot \left( \frac{4}{\epsilon} \right)^D.
\]
\end{lemma}
\begin{proof}
Fix an $i\in[1, \ell]$. We prove an upper bound on the number $\tau_i$
of iterations of \gmm needed to obtain $r_{T_i^{\tau_i}}(S_i)
\le (\epsilon/2) r_{T_i^k}(S_i)$, which in turn 
bounds the size of $T_i$.
Consider the $k$-center clustering of $S_i$ induced by the
$k$ centers in $T_i^k$, with radius
$r_{T_i^k}(S_i)$.  By the doubling dimension property, we have that
each of the $k$ clusters can be covered using at most $(4/\epsilon)^D$
balls of radius $\le (\epsilon/4) \cdot r_{T_i^k}(S_i)$, for
a total of at most  $h=k(4/\epsilon)^D$ such balls.
Consider now the execution of $h$ iterations
of the \gmm algorithm on $S_i$. Let $T_i^h$ be the set
of returned centers and let $x \in S_i$ be the farthest point
of $S_i$ from $T_i^h$. The center selection process of the 
\gmm algorithm ensures that any two points in
$T_i^h \cup \{x\} $ are at distance at least
$r_{T_i^h}(S_i)$ from one another. Thus, since two of these points
must fall into one of the $h$ aforementioned balls of 
radius $\le (\epsilon/4) \cdot r_{T_i^k}(S_i)$, 
this implies immediately (by the triangle inequality) that
\[
r_{T_i^h}(S_i)  
\le 2 (\epsilon/4) \cdot r_{T_i^k}(S_i)
= (\epsilon/2) \cdot r_{T_i^k}(S_i).
\]
Hence,  after
$h$ iterations we are guaranteed that \gmm finds a set 
$T_i^h$ which meets the stopping condition. Therefore, 
$|T_i|=\tau_i \leq h = k(4/\epsilon)^D$, for every $i \in [1,\ell]$,
and the lemma follows.
\end{proof}
We now state the main result of this subsection.
\begin{theorem}
Let $0 < \epsilon \le 1$.
If the points of $S$ belong to a metric space of doubling dimension
$D$, then the above 2-round MapReduce algorithm computes a
$(2+\epsilon)$-approximation for the $k$-center problem 
with local memory 
$M_L = \BO{|S|/\ell + \ell\cdot k \cdot (4/\epsilon)^D}$
and linear aggregate memory.
\end{theorem}
\begin{proof}
Let $X$ be the solution found by \gmm on $T$.  Since $T \subseteq S$,
from Lemma~\ref{lem:gmm-subset} it follows that $r_X(T) \le 2 \cdot
\opt$.  Consider an arbitrary point $s \in S$, along with its proxy
$p(s) \in T$, as defined before.  By Lemma~\ref{lem:k-center-proxy} we
know that $d(s, p(s)) \le \epsilon\cdot\opt$.  Let $x\in X$ be the
center closest to $p(s)$.  It holds that $d(x, p(s)) \le 2\cdot \opt$.
By applying the triangle inequality, we have that $d(x, s) \le d(x,
p(s)) + d(p(s), s) \le 2\cdot \opt + \epsilon\cdot \opt =
(2+\epsilon)\opt$. The bound on $M_L$ follows
since in the first round each processor needs
to store $|S|/\ell$ points of the input and computes a coreset of size
$\BO{k \cdot (4/\epsilon)^D}$, as per Lemma~\ref{lem:k-center-size},
while in the second round, one processor needs enough memory to store
$\ell$ such coresets. Finally, it is immediate to see that
aggregate memory proportional to the input size
suffices. 
\end{proof}
By setting $\ell = \BT{\sqrt{|S|/k}}$ in the above theorem 
we obtain:
\begin{corollary}
Our 2-round MapReduce algorithm computes a
$(2+\epsilon)$-approximation for the $k$-center problem 
with local memory $M_L = \BO{\sqrt{|S|k} (4/\epsilon)^D}$ and linear
aggregate memory. For constant $\epsilon$ and $D$, 
the local memory bound becomes $M_L = \BO{\sqrt{|S|k}}$.
\end{corollary}

\subsection{MapReduce algorithm for $k$-center with $z$ outliers} \label{sec-out}

Consider an instance $S$ of the $k$-center problem with $z$ outliers
and fix a precision parameter $\hat{\epsilon} \in (0, 1]$ intended, as before,
to regulate the approximation ratio. 
We propose the following
  2-round MapReduce algorithm for the problem. In the first round, $S$
  is partitioned into $\ell$ equally-sized subsets $S_i$, with $1 \leq
  i \leq \ell$, and for each $S_i$, in parallel, \gmm is run
  incrementally. Let $T_i^j$ be the set of the first $j$ selected
  centers.  We continue to run \gmm
  until the first iteration $\tau_i \geq k+z$ such that
  $r_{T_i^{\tau_i}}(S_i) \le \hat{\epsilon}/2 \cdot
  r_{T_i^{k+z}}(S_i)$.  Define the coreset $T_i=T_i^{\tau_i}$. As before,
  for each point $s \in S_i$ we define its \emph{proxy} $p(s)$
  to be the point of $T_i$ closest to $s$, but, furthermore, 
we attach to each
  $t \in T_i$ a \emph{weight} $w_t \geq 1$, which is the number of
  points of $S_i$ with proxy $t$.

In the second round, the union of the weighted coresets  $T = \cup_{i=1}^\ell
T_i$ is gathered into a single reducer. Before describing the details of
this second round, we need to introduce
a sequential algorithm, dubbed \outliers\  (see pseudocode below),
for solving a weighted variant of the $k$-center problem with outliers
which is a modification of the one presented in \cite{MalkomesKCWM15}
(in turn, based on the unweighted algorithm 
of \cite{Charikar2001}). 
\begin{algorithm}
  \caption{\outliers{}$(T, k, r, \hat{\epsilon})$}
  \label{algo:outliers-seq}
  \DontPrintSemicolon

  \Let{$T'$}{$T$}\;
  \Let{$X$}{$\emptyset$}\;

  \While{$((|X| < k) \; \mbox{\bf and} \; (T' \neq \emptyset))$}{
    \lFor{\rm ($t \in T$)} {
      \Let{$B_t$}{$\{v : v \in T' \wedge d(v, t) \le (1+2\hat{\epsilon})\cdot r \}$}
    }
    \Let{$x$}{$\argmax_{t \in T} \sum_{v \in B_t} w_v$}\;
    \Let{$X$}{$X \cup \{x\}$}\;
    \Let{$E_x$}{$\{v : v \in T' \wedge d(v, x) \le (3+4\hat{\epsilon})\cdot r \}$}\;
    \Let{$T'$}{$T' \setminus E_x$}\;
  }
  \Return $X, T'$\;
\end{algorithm}

\noindent
\outliers$(T,k,r,\hat{\epsilon})$ returns two subsets $X, T' \subseteq
T$ such that $X$ is a set of (at most) $k$ centers, and $T'$ is a set
of points referred to as \emph{uncovered points}. The algorithm starts
with $T'=T$ and builds $X$ incrementally in $|X| \le k$ iterations as
follows.  In each iteration, the next center $x$ is chosen as the point
maximizing the aggregate weight of uncovered points in its ball of radius
$(1+2\hat{\epsilon})\cdot r$ (note that $x$ needs not be an uncovered
point).  Then, all uncovered points at distance at most
$(3+4\hat{\epsilon})\cdot r$ from $x$ are removed from $T'$.  The
algorithm terminates when either $|X|=k$ or $T'=\emptyset$. By
construction, the final $T'$ consists of all points at distance
greater than $(3+4\hat{\epsilon})\cdot r$ from $X$.

\sloppy Let us return to the second round of our MapReduce
algorithm. The reducer that gathered $T$ runs
\outliers{}$(T,k,r,\hat{\epsilon})$ multiple times to estimate the
minimum value $r_{\rm min}$ such that the aggregate weight of the
points in the set $T'$ returned by \outliers{}$(T,k,r_{\rm
  min},\hat{\epsilon})$ is at most $z$. More specifically, the
computed estimate, say $\tilde{r}_{\rm min}$, is within a
multiplicative tolerance $(1+\delta)$ from the true $r_{\rm min}$,
with $\delta = \hat{\epsilon}/(3+4\hat{\epsilon})$, and it is obtained
through a binary search over all possible $\BO{|T|^2}$ distances
between points of $T$ combined with a geometric search with step
$(1+\delta)$. 
To avoid storing all $\BO{|T|^2}$ distances, the value of $r$ at each
iteration of the binary search can be determined in space linear in
$T$ by the median-finding Streaming algorithm in \cite{MunroP80}.
The output of the MapReduce algorithm is the set of centers
computed by \outliers{}$(T,k,\tilde{r}_{\rm min},\hat{\epsilon})$.

We now analyze our 2-round MapReduce algorithm. The following lemma
bounds the distance between a point and its proxy.
\begin{lemma}\label{lem:proxy-outliers}
For each $s \in S$, 
$d(s, p(s)) \le \hat{\epsilon}\cdot\opto$.
\end{lemma}
\begin{proof}
Consider any subset $S_i$ of the partition $S_1, \dots, S_\ell$ of
$S$.  By construction, we have that for each $s \in S_i$,
$d(s, p(s)) \le
(\hat{\epsilon}/2)\cdot r_{T_i^{k+z}}(S_i)$.  Since $S_i$ is a subset of
$S$, Lemma~\ref{lem:gmm-subset} ensures that $r_{T_i^{k+z}}(S_i) \le 2
r^*_{k+z}(S)$.  Hence, $d(s, p(s)) \le \hat{\epsilon} r^*_{k+z}(S)$.
Since  $r^*_{k+z}(S) \le  \opto$, as observed before in Eq.~\ref{eq:radius-relation},
we have $d(x, p(x)) \le
\hat{\epsilon}\cdot\opto$.
\end{proof}

Next, we characterize the quality of the solution returned
by \outliers{} when run on  $T$, the union of the weighted coresets,
and with a radius $r\geq \opto$.

\begin{lemma}\label{lem:outliers-approx}
\sloppy
For $r\geq \opto$, let $X,T' \subseteq T$ be the sets returned by
\outliers$(T,k,r,\hat{\epsilon})$, and define $S_{T'} = \{s\in S: p(s)\in T'\}$. Then,
\[
 d(t, X) \le (3+4\hat{\epsilon}) \cdot r
\quad \forall t \in T \setminus T'
  \]
and  $|S_{T'}| \le z$.
\end{lemma}
\begin{proof}
\sloppy
The proof uses an argument akin to the one used for the
analysis of the sequential algorithm by \cite{Charikar2001} and later
adapted by \cite{MalkomesKCWM15} to the weighted coreset setting.
The first claim follows immediately from the workings of the algorithm, since 
each point in $T-T'$ belongs to some $E_x$, with $x\in X$. 
We are left to show that $|S_{T'}| \le z$.
Suppose first that $|X|<k$. In this case, it must be $T' = \emptyset$,
hence  $|S_{T'}|= 0$, and the proof follows. We now
 concentrate on the case $|X|=k$. Consider the $i$-th
iteration of the while loop of \outliers$(T,k,r,\hat{\epsilon})$
and define $x_i$ as the center of $X$ selected in the iteration, and 
$T'_i$ as the set $T'$ of uncovered points at the beginning
of the iteration. Recall that $x_i$ is the
point of $T$  which maximizes the cumulative weight of
the set $B_{x_i}$ of uncovered points in $T'_i$  at distance at most
$(1+2\hat{\epsilon})\cdot r$ from $x_i$, and that
the set $E_{x_i}$ of all uncovered points at distance 
at most $(3+4\hat{\epsilon})\cdot r$ from $x_i$ is removed from
$T'_i$  at the end of the iteration. 
We now show that 
\begin{equation} \label{eq:weight-outliers}
\sum_{i=1}^k \sum_{t \in E_{x_i}} w_t \ge |S|-z,
\end{equation}
which will immediately imply that $|S_{T'}| \le z$.
For this purpose, let $O$ be an optimal set of $k$ centers for 
the problem instance under consideration, and let
$Z$ be the set of at most $z$ outliers at distance 
greater than $\opto$ from $O$.
For each $o \in O$,
define 
$C_o \subseteq S \setminus Z$ as the set of nonoutlier points which
are closer to $o$ than to any other center of $O$, with ties
broken arbitrarily. To prove~(\ref{eq:weight-outliers}), 
it is sufficient
to exhibit an ordering
$o_1, o_2, \ldots, o_k$ of the centers in $O$ so that,
for every $1 \leq i \leq k$, it holds
\[
\sum_{j=1}^i \sum_{t \in E_{x_j}} w_t \ge |C_{o_1} \cup \dots \cup C_{o_i}|.
\]
The proof uses an inductive
charging argument to assign each point in 
$\bigcup_{j=1}^i C_{o_j}$ 
to a point in $\bigcup_{j=1}^i E_{x_j}$, where
each $t$ in the latter set will be in charge of at most
$w_t$ points.  We define two charging rules. A point can be
either charged to its own proxy (\emph{Rule 1}) or to another point
of $T$ (\emph{Rule 2}).

Fix some arbitrary $i$, with $1 \leq i \leq k$,
and assume, inductively, that the points in 
$C_{o_1} \cup \dots \cup C_{o_{i-1}}$ 
have been charged to points in $\bigcup_{j=1}^{i-1} E_j$
for some choice of distinct optimal centers
$o_1, o_2, \ldots, o_{i-1}$. We have  two cases. \\
{\bf Case 1.} \emph{There exists an optimal center $o$ still
unchosen such that 
there is a point $v\in C_{o}$ with $p(v) \in B_{x_j}$, 
for some $1 \leq j \leq i$.} We choose $o_i$ as one such center.
Hence $d(x_j,p(v)) \le (1 + 2\hat{\epsilon}) \cdot r$.
By repeatedly applying the triangle inequality 
we have that
for each $u \in C_{o_i}$
\begin{align*}
d(x_j,p(u)) \leq &
\;\; d(x_j,p(v)) + 
d(p(v),v) + 
d(v, o_i) +
d(o_i,u) + \\
& + d(u,p(u))  \leq (3 + 4\hat{\epsilon}) \cdot r
\end{align*}
hence,  $p(u) \in E_{x_j}$. Therefore
we can charge each point $u\in C_{o_i}$ to its proxy, by Rule
1. \\
{\bf Case 2.}
\emph{For each unchosen optimal center $o$ and each
$v \in C_o$, $p(v) \not\in \bigcup_{j=1}^i B_{x_j}$.}
We choose $o_i$ to be the unchosen optimal center which
maximizes the cardinality of 
$\{p(u) : u \in C_{o_i}\} \cap T'_i$.
We distinguish between points $u\in C_{o_i}$ with $p(u) \notin T'_i$,
hence $p(u) \in \bigcup_{j=1}^{i-1} E_{x_j}$, and those with
$p(u) \in T'_i$.  We charge each $u\in C_{o_i}$ with $p(u) \notin
T'_i$ to its own proxy by Rule 1.  As for the other points,
we now show that we can charge them to the points of
$B_{x_i}$.  To this purpose, we first observe that
$B_{p(o_i)}$ contains 
$\{p(u) : u \in C_{o_i}\} \cap T'_i$, since
for each $u \in C_{o_i}$ 
\[
\begin{aligned}
  d(p(o_i),p(u)) 
  &\leq 
  d(p(o_i),o_i) + 
  d(o_i,u) + 
  d(u,p(u)) \\
  &\leq (1 + 2\hat{\epsilon}) \cdot \opto \leq (1 + 2\hat{\epsilon}) \cdot r.
\end{aligned}
\]


%
Therefore the aggregate weight of $B_{p(o_i)}$ is at least 
$\left|\left\{u \in C_{o_i} : p(u) \in T'_i\right\}\right|$.
Since Iteration $i$ selects $x_i$ as the 
center such that $B_{x_i}$ has maximum aggregate weight,
we have that
\[
\sum_{t \in B_{x_i}} w_t 
\ge \sum_{z \in B_{p(o_i)}} w_z
\ge \left|\left\{u \in C_{o_i} : p(u) \in T'_i\right\}\right|,
\]
hence, the points in $B_{x_i}$ have enough weight to be charged with
each point $u \in C_{o_i}$ with $p(u) \in T'_i$.  
Figure~\ref{fig:case-2-charging-rule} illustrates the charging under Case 2.
\begin{figure}
  \centering
  \begin{tikzpicture}[scale=1.3]
    \tikzstyle{point}=[circle, fill=black, minimum size=5pt, inner sep=0pt]
    \tikzstyle{proxy}=[star,star points=5, fill=black, minimum size=8pt, inner sep=0pt]

    \def\centerarc(#1)(#2:#3:#4){($(#1)+({#4*cos(#2)},{#4*sin(#2)})$) arc (#2:#3:#4) }

    \def\proxypoint[#1](#2){
      (#2)node[proxy,#1]{} 
      \centerarc(#2)(0:360:0.3)
    }

    \node (E) at (-3,0) {};
    \node[point,label={[label distance=0pt]above:$o_i$}] (O) at (0.3,1) {};
    \node[proxy, label=$x_i$] (X) at (2.5,-0.5) {};

    \def\largedisk{\centerarc(E)(45:-30:3)}
    \def\optimaldisk{\centerarc(O)(0:360:1)}
    \def\greedydisk{\centerarc(X)(0:360:1.2)}

    \tikzfading[
      name=fade custom,
      top color=transparent!100,
      bottom color=transparent!100,
      right color=transparent!0,
    ]



    \draw[dashed] \optimaldisk; 
    \draw[dashed] \largedisk;
    \draw[dashed] \greedydisk ;

    \node[point] (covered) at (0.08,0.8) {};
    \node[point] (covered2) at (-0.3,0.3) {};
    \node[point] (uncovered) at (0.9,1.2) {};
    \node[point] (uncovered2) at (0.6,1.85) {};

    \node (proxycovered) at (-0.22, 0.8) {};
    \draw[dotted] \proxypoint[](proxycovered); 
    \draw[-{Latex[width=.8ex, length=1ex]}] 
      (covered) .. controls (-0.12,0.4) and (-0.15,0.5) .. (proxycovered);

    \node (proxycovered2) at (-0.5, 0.1) {};
    \draw[dotted] \proxypoint[](proxycovered2); 
    \draw[-{Latex[width=.8ex, length=1ex]}] 
      (covered2) .. controls (-0.2, 0.05) and (-0.3,0) .. (proxycovered2);

    \node (proxyuncovered) at (1,1.3) {};  
    \draw[dotted] \proxypoint[](proxyuncovered);       
    \draw[-{Latex[width=.8ex, length=1ex]}] (uncovered) arc (80:30:3);

    \node (proxyuncovered2) at (0.65,2.1) {};  
    \draw[dotted] \proxypoint[](proxyuncovered2);    
    \draw[-{Latex[width=.8ex, length=1ex]}] 
      (uncovered2) arc (175:250:3);

    \draw (1.5,2.1) node[right]{
      \parbox{8em}{\small points with their proxy not covered are charged to $B_{x_i}$ by Rule 2}
      };
    \draw (-0.7,-0.8) node[left]{\parbox{8em}{\small points with their proxy covered by $E_{x_j}$, for some $j < i$, are charged to their proxy by Rule 1}};


    \node[label=$E_{x_j}$] at (.1,-1.3) {};
    \node[label=$B_{x_i}$] at (1.15,-1) {};
    \node[label=$C_{o_i}$] at (0.7,-0.4) {};

  \end{tikzpicture}
  \caption{Application of charging rules in case 2 of the proof. Round points are points of $S$, whereas star-shaped points are proxy points in $T$. Arrows represent charging.}
  \label{fig:case-2-charging-rule}
\end{figure}
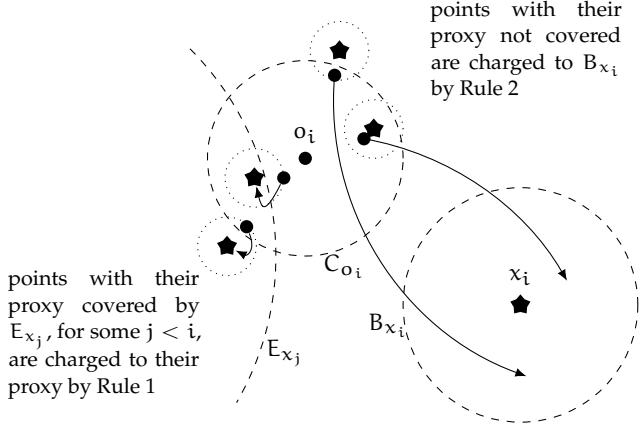
Note that the points of $B_{x_i}$ did not receive any charging by Rule
1 in previous iterations, since they are uncovered at the beginning of
Iteration $i$, 
and will not
receive chargings by Rule 1 in subsequent iterations, since $B_{x_i}$
does not intersect the set $C_o$ of any optimal center $o$ yet to be
chosen.  Also, no further charging to points of $B_{x_i}$ by Rule 2
will happen in subsequent iterations, since Rule 2
will only target sets $B_{x_h}$ with $h > i$. These observations
ensure that any point of $T$ receives charges through either Rule 1 or
Rule 2, but not both, and never in excess of its weight, and the proof
follows.
\end{proof}

The following lemma bounds the size of $T$, the union of
the weighted coresets.
\begin{lemma}\label{lem:outliers-size}
If $S$ belongs to a metric space of doubling dimension
$D$, then 
\[
|T| \le \ell\cdot (k+z) \cdot \left(
\frac{4}{\hat{\epsilon}}  
\right)^D
\]
\end{lemma}
\begin{proof}
The proof proceeds similarly to the one of
Lemma~\ref{lem:k-center-size}, with the understanding that the
definition of doubling dimension is applied to each of the $(k+z)$
clusters induced by the points of $T_i^{k+z}$ on $S_i$.
\end{proof}

Finally, we 
state the main result of this subsection.
\begin{theorem}\label{thm:outliers-mr}
Let $0 < \epsilon \le 1$.
If the points of $S$ belong to a metric space of doubling dimension
$D$, then, when run with $\hat{\epsilon}=\epsilon/6$,
the above 2-round MapReduce algorithm
computes a
$(3+\epsilon)$-approximation for the $k$-center problem 
with $z$ outliers with local memory
$M_L = \BO{|S|/\ell + \ell\cdot (k+z) \cdot (24/\epsilon)^D}$
and linear aggregate memory.
\end{theorem}

\begin{proof}
The result of  Lemma~\ref{lem:outliers-approx} 
combined with the stipulated tolerance of the
search
performed in the second round of the 
algorithm implies that the radius discovered by the search 
is $\tilde{r}_{\rm min} \le \opto (1+\delta)$ with $\delta=\hat{\epsilon}/(3+4\hat{\epsilon})$.
Also, by the triangle inequality, the distance between each non-outlier 
point in $S$ and its closest center will be at most
$\hat{\epsilon}\opto + (3+4\hat{\epsilon}) \opto(1+\delta) \leq
 (3+6 \hat{\epsilon}) \opto \leq  (3+\epsilon)\opto$,
which proves the approximation bound. 
The bound on $M_L$ follows since
in the first round each reducer needs enough memory
to store $|S|/\ell$ points of
the input, while in the second round
the
reducer computing the final solution requires enough 
memory to store the union of the $\ell$ coresets, 
which, by Lemma~\ref{lem:outliers-size}, has
size $\BO{(k+z)(4/\hat{\epsilon})^D}
= \BO{(k+z)(24/\epsilon)^D}$ each.  Also,
globally, the reducers need only sufficient
memory to store the input, hence $M_A = \BO{|S|}$.
\end{proof}

By setting $\ell = \BT{\sqrt{|S|/(k+z)}}$ in the above theorem 
we obtain:
\begin{corollary}
\sloppy
Our 2-round MapReduce algorithm computes a
$(3+\epsilon)$-approximation for the $k$-center problem with
$z$ outliers, 
with local memory $M_L = \BO{\sqrt{|S|(k+z)}(24/\epsilon)^D)}$ and linear
aggregate memory. For constant $\epsilon$ and $D$, 
the local memory bound becomes $M_L = \BO{\sqrt{|S|(k+z)}}$.
\end{corollary}
\sloppy {\bf Improved sequential algorithm.}  A simple analysis implies that, by setting
$\ell=1$, our MapReduce strategy for the $k$-center problem with $z$
outliers yields an efficient sequential $(3+\epsilon)$-approximation
algorithm whose running time is $\BO{|S||T|+k |T|^2 \log |T|}$, where
$|T|=(k+z)(24/\epsilon)^D$, is the coreset size. For a wide range of
values of $k,z, \epsilon$ and $D$ this yields a substantially improved
performance over the $\BO{k |S|^2 \log |S|}$-time state-of-the-art
algorithm of \cite{Charikar2001}, at the expense of a negligibly worse
approximation.

\subsubsection{Higher space efficiency through randomization}
\label{sec:improved-space}
The analysis of very noisy datasets might 
require setting the number $z$ of outliers much larger than $k$, while still 
$o(|S|)$. In this circumstance, the size of the 
union of the coresets
$T$ is proportional to $\sqrt{|S|z}$, and may turn out too large for practical
purposes, due to the large local memory requirements and to the
running time of the cubic sequential approximation algorithm run on
$T$ in the second round, which may become the real performance
bottleneck of the entire algorithm.  In this subsection, we show that
this drawback can be significantly ameliorated by simply partitioning
the pointset at random in the first round, at the only expense of
probabilistic rather than deterministic guarantees on the resulting
space and approximation guarantees. We say that an event
related to a dataset $S$ occurs \emph{with high probability} $p$ if 
$p \geq 1-1/|S|^c$, for some constant $c \geq 1$.

The randomized variant of the algorithm works as follows.  In the first
round, the input set $S$ is partitioned into $\ell$ subsets $S_i$,
with $1 \leq i \leq \ell$, by assigning each point to a random subset
chosen  uniformly and  independently of the other points. Let
$z'=6((z/\ell)+\log_2 |S|)$ and observe that, for large $z$ and $\ell$,
we have that $z' \ll z$.  Then, in parallel on each partition $S_i$, \gmm is run to
yield a set $T_i^{\tau_i}$ of $\tau_i$ centers, where $\tau_i \geq
k+z'$ is the smallest value such that $r_{T_i^{\tau_i}}(S_i) \le
(\hat{\epsilon}/2) \cdot r_{T_i^{k+z'}}(S_i)$.  Define the coreset
  $T_i=T_i^{\tau_i}$ and, again, for each point $s \in S_i$  define
  its \emph{proxy} $p(s)$ to be the point of $T_i$ closest to $s$.
The rest of the algorithm is exactly as before using these new
   $T_i$'s.

The analysis proceeds as follows.  Consider an optimal solution of the
$k$-center problem with $z$ outliers for $S$, and let $O=\{o_1, o_2, \ldots,
o_k\}$ be the set of $k$ centers and $Z_O$ the set of $z$ outliers, that
is the $z$ points of $S$ most distant from $O$. Recall that any point
of $S\setminus Z_O$ is at distance at most $\opto$ from $O$.
The following lemma states that the outliers (set $Z_O$) are well
distributed among the $S_i$'s.

\begin{lemma}\label{lem:occupancy}
With high probability, each $S_i$ contains
no more than $z'=6((z/\ell)+\log_2 |S|)$ points of $Z_O$.
\end{lemma}
\begin{proof}
The result follows by applying   Chernoff bound  (4.3) of 
\cite{MitzemacherU17} and the union bound, which yield
that the stated event occurs with probability at least $1-1/|S|^5$.
\end{proof}

The rest of the analysis mimics
the one of the deterministic version.
\begin{lemma}\label{lem:condensed}
The statements of both Lemmas~\ref{lem:proxy-outliers}
and~\ref{lem:outliers-approx}
hold with high probability. 
\end{lemma}
\begin{proof}
We first prove that, with high probability, for each for each $s \in
S$, $d(s, p(s)) \le \hat{\epsilon}\cdot\opto$ (same as
Lemma~\ref{lem:proxy-outliers}). Consider $O$ and $Z_O$.
We condition on the event that each $S_i$ contains at most $z'$
points of $Z_O$, which, by Lemma~\ref{lem:occupancy}, 
occurs with high probability. Focus on an arbitrary subset $S_i$.
For $1 \leq j \leq \ell$, let $C_j$ be the set of points of
$S\setminus Z_O$ whose closest optimal center is $o_j$, and let
$C_j(i) = C_j \cap S_i$. Consider the set $T_i^{k+z'}$ of centers
determined by the first $k+z'$ iterations of the \gmm algorithm
and let $x \in S_i$ be the farthest point of $S_i$ from
$T_i^{k+z'}$. By arguing as in the proof of Lemma~\ref{lem:k-center-size},
it can be shown that any two points in
$T_i^{k+z'} \cup \{x\}$ are at distance at least
$r_{T_i^{k+z'}}(S_i)$ from one another and since two of these
points must belong to the same $C_j(i)$ for some $j$,
by the triangle inequality we have that   
\[
r_{T_i^{k+z'}}(S_i) \le 2 \opto.
\]
Recall that the \gmm algorithm on $S_i$ is stopped
at the first iteration $\tau_i$ such that
$r_{T_i^{\tau_i}}(S_i) \leq
(\hat{\epsilon}/2) \cdot r_{T_i^{k+z'}}(S_i)$, hence
\[
r_{T_i^{\tau_i}}(S_i) \leq 
(\hat{\epsilon}/2) \cdot r_{T_i^{k+z'}}(S_i) \leq
(\hat{\epsilon}/2) \cdot 2 \opto =
\hat{\epsilon} \cdot \opto.
\]
The desired bound on $d(s,p(s))$ immediately follows. 
Conditioning on this bound, the proof of
Lemma~\ref{lem:outliers-approx} can be repeated identically,
hence the stated property holds. 
\end{proof}
By repeating the same argument used in Lemma~\ref{lem:outliers-size}, one can easily
argue that, if $S$ belongs to a metric space of doubling dimension
$D$, then 
the size of the weighted coreset $T$ is
\[
|T| \le \ell\cdot (k+z') \cdot \left(
\frac{4}{\hat{\epsilon}}  
\right)^D.
\]
\sloppy
This bound, together with the results of the preceding lemma,
immediately implies the analogous of Theorem~\ref{thm:outliers-mr}
stating that, with high probability, the randomized algorithm computes
a $(3+\epsilon)$-approximation for the $k$-center problem with $z$
outliers with local memory $M_L = \BO{|S|/\ell + \ell\cdot (k+z')
  \cdot (24/\epsilon)^D}$ and linear aggregate memory.  Observe that
$z$ is now replaced by (the much smaller) $z'$ in the local memory
bound. 

By choosing $\ell = \BT{\sqrt{|S|/(k+\log |S|)}}$ we obtain:
\begin{corollary}
\sloppy
With high probability, our 2-round MapReduce algorithm computes a
$(3+\epsilon)$-approximation for the $k$-center problem with
$z$ outliers, 
with local 
memory
$M_L = \BO{\left(\sqrt{|S|(k+\log |S|)}+z\right)(24/\epsilon)^D}$ and linear
aggregate memory. For constant $\epsilon$ and $D$, 
the local memory bound becomes  $M_L = \BO{\sqrt{|S|(k+\log |S|)}+z}$ 
\end{corollary}
With respect to the deterministic version, 
for large values of $z$ a substantial improvement
in the local memory requirements is achieved. \\[0.1cm]
{\bf Remark.} Thanks to the incremental nature of \gmm, our
coreset-based MapReduce algorithms for the $k$-center problem, both
without and with outliers, need not know the doubling dimension $D$ of
the underlying metric space in order to attain the claimed performance
bounds. This is a very desirable property, since, in general, $D$ may
not be known in advance. Moreover, if $D$ were known,
a factor $\sqrt{(c/\epsilon)^D}$ in local memory
(where $c=4$ for $k$-center, and $c=24$ for $k$-center with
$z$ outliers) could be saved
by setting  $\ell$ to be a factor
$\BT{\sqrt{(c/\epsilon)^D}}$ smaller.

\section{Streaming algorithm for $k$-center with $z$ outliers}
\label{sec:streaming}

As mentioned in the introduction, in the Streaming setting we will
only consider the $k$-center problem with $z$ outliers. Consider an
instance $S$ of the problem and fix a precision parameter
$\hat{\epsilon} \in (0, 1]$.  Suppose that the points of $S$ belong to
  a metric space of known doubling dimension $D$.  Our Streaming
  algorithm also adopts a coreset-based approach.  Specifically, in a
  pass over the stream of points of $S$ a suitable weighted coreset
  $T$ is selected and stored in the working memory. Then, at the end
  of the pass, the final set of centers is determined through multiple runs of
  \outliers{} on $T$ as was done in the second round of the MapReduce
  algorithm described in Subsection~\ref{sec-out}.  Below, we will
  focus on the coreset construction.

The algorithm computes a coreset $T$ of $\tau \geq k+z$ centers which
represent a good approximate solution to the $\tau$-center problem on
$S$ (without outliers).  The value of $\tau$, which will be fixed
later, depends on $\hat{\epsilon}$ and $D$.  The main difference with
the MapReduce algorithm is the fact that we cannot exploit the
incremental approach provided by \gmm, since no efficient
implementation of \gmm in the Streaming setting is known. Hence, for
the computation of $T$ we resort to a novel weighted variant of the
\emph{doubling algorithm} by Charikar et al.~\cite{CharikarCFM04}
which is described below.

For a given stream of points $S$ and a target number of centers
$\tau$, the algorithm maintains a weighted set $T$ of centers selected among
the points of $S$ processed so far, and a lower bound $\phi$ on
$r_\tau^*(S)$. $T$ is initialized with the first $\tau+1$ points of $S$,
with each $t \in T$  assigned weight $w_t=1$,
while $\phi$ is initialized to half the minimum distance between
the points of $T$.
For the sake of the analysis, we
will define a proxy function $p: S \to T$ which, however, will not be
explicitly stored by the algorithm. Initially, each point of $T$ is
proxy for itself. The remaining points of $S$ are processed
one at a time maintaining the following invariants:
\begin{enumerate}[label=(\alph*)]
\item $T$ contains at most $\tau$ centers.
\item $\forall t_1, t_2 \in T$ we have $d(t_1, t_2) > 4\phi$
\item $\forall s \in S$ processed so far, 
$d(s,p(s)) \le 8\phi$.
\item $\forall t\in T$, $w_t = |\{s \in S \mbox{ processed so far}
: \; p(s) = t\}|$.
\item $\phi \le r_\tau^*(S)$.
\end{enumerate} 
The following two rules are applied to process each
new point $s \in S$.  The \emph{update rule} checks if $d(s,T) \le
8\phi$. If this is the case, the center $t \in T$ closest to $s$ is
identified and $w_t$ is incremented by one, defining $p(s) = t$. If
instead $d(s,T) > 8\phi$, then $s$ is added as a new center to $T$,
setting $w_s$ to 1 and defining $p(s)=s$.  Note that in this latter
case, the size of $T$ may exceed $\tau$, thus violating invariant (a).
When this happens, the following \emph{merge rule} is invoked
repeatedly until invariant (a) is re-established.  Each invocation of
this rule first sets $\phi$ to $2\phi$, which, in turn, may lead to a
violation of invariant (b).  If this is the case, for each pair of
points $u, v \in T$ violating invariant (b), we discard $u$ and set
$w_v \leftarrow w_v + w_u$. Conceptually, this corresponds to the
update of the proxy function which redefines $p(x)=v$, for each point
$x$ for which $p(x)$ was equal to $u$. 

Observe that, at the end of the initialization, invariants (a) and (b)
do not hold, while invariants (c)$\div$(e) do hold. Thus, we prescribe
that the merge rule and the reinforcement of invariant (b) are applied
at the end of the initialization before any new point is
processed. This will ensure that all invariants hold before the
$(\tau+2)$nd point of $S$ is processed. The following lemma shows the
above rules maintain all invariants.

\begin{lemma}\label{lem:doubling-algorithm}
After the initialization, at the end of the processing of each point $s \in S$,
all invariants hold. 
\end{lemma}
\begin{proof}
As explained above, all invariants are enforced at the end of the
initialization.  Consider the processing of a new point $s$. It is
straightforward to see that the combination of update and merge rules
maintain invariants (a)-(d). We now show that invariant (e) is also
maintained.  After the update rule is applied, only invariant (a) can
be violated.  Suppose that this is the case, hence $|T| = \tau+1$.
Each pair of centers in $T$ are at distance at least $4\phi$ from one
another (invariant (b)). Let $\phi'$ be the new value of $\phi$
resulting after the required applications of the merging rule.  It is
easy to see that until the penultimate application of the merge rule,
$T$ still contains $\tau+1$ points. Therefore each pair of these
points must be at distance at least $4(\phi'/2)=2\phi'$ from one
another. This implies, that $\phi'$ is still a lower bound to
$r_\tau^*(S)$.
\end{proof}

As an immediate corollary of the previous lemma, we have that after
all points of $S$ have been processed,
$d(s, p(s)) \le 8 \cdot r_\tau^*(S)$ for every $s \in S$.
Moreover, it is immediate to see that the working memory 
required by the algorithm has size $\BT{\tau}$.
Fix now
$\tau = (k+z)(16/\hat{\epsilon})^D$ 
and let $T$ be the weighted coreset $T$ of size $\tau$
returned by the above algorithm. The following lemma
 is 
the counterpart of Lemma~\ref{lem:proxy-outliers}
in the Streaming setting.
\begin{lemma}\label{lem:proxy-outliers-streaming}
For every $s \in S$, $d(s, p(s)) \le \hat{\epsilon}\cdot \opto$. 
\end{lemma}
\ifextended
\begin{proof}
Observe that $S$ can be covered using $k+z$ balls of radius
$r_{k+z}^*(S)$.  Since $S$ comes from a space of doubling dimension
$D$, we know that $S$ can also be covered using $\tau =
(k+z)(16/\hat{\epsilon})^{D}$ balls (not necessarily centered at points
in $S$) of radius $\le\hat{\epsilon}/16 \cdot r_{k+z}^*(S)$.  
Picking an arbitrary center of $S$ from each such ball induces
a $\tau$-clustering of $S$ with radius at most $\hat{\epsilon}/8 \cdot
r_{k+z}^*(S)$. Hence,
\[
r_\tau^*(S) \le \hat{\epsilon}/8 \cdot r_{k+z}^*(S).
\]
Since $r_{k+z}^*(S) \le
\opto$, it follows that $r_\tau^*(S) \le \hat{\epsilon}/8 \cdot \opto$.
By invariants (c) and (e) we have that for every $s \in S$
\[
d(s, p(s)) \le 8 \phi \le 
8\cdot r_\tau^*(S) \le \hat{\epsilon} \cdot \opto.
\qed
\]
\end{proof}
\fi
The following theorem states the main result of this section.
\begin{theorem}\label{thm:streaming-outliers}
Let $0 < \epsilon \le 1$.
If the points of $S$ belong to a metric space of doubling dimension
$D$, then, when run with $\hat{\epsilon}=\epsilon/6$,
the above 1-pass Streaming algorithm
computes a
$(3+\epsilon)$-approximation for the $k$-center problem 
with $z$ outliers with working memory of size
$\BO{(k+z)(96/\epsilon)^D}$.
\end{theorem}
\begin{proof}
Given the result of Lemma~\ref{lem:proxy-outliers-streaming},
the approximation factor can be established in exactly the same way as
done for the MapReduce algorithm (refer to Lemma~\ref{lem:outliers-approx} and 
Theorem~\ref{thm:outliers-mr}), while the 
bound on the working memory size follows directly from the choice of $\hat{\epsilon}$,
the fact that $|T|=\tau=(k+z)(16/\hat{\epsilon})^D$, and 
the fact that the Streaming algorithm needs memory proportional  $|T|$.
\end{proof}

\begin{corollary}\label{corol:streaming-outliers}
For constant $\epsilon$ and $D$, 
the above Streaming algorithm computes a
$(3+\epsilon)$-approximation for the $k$-center problem with
$z$ outliers with working memory of size
$\BO{(k+z)}$, independent of $|S|$.
\end{corollary}

A few remarks are in order. For simplicity, to compute the weighted
coreset $T$ we preferred to adapt the 8-approximation algorithm by
\cite{CharikarCFM04} rather than the more complex 
$(2+\epsilon)$-approximation algorithm by \cite{McCutchen2008}, since
this choice does not affect the approximation guarantee of our
algorithm but comes only at the expense of a slight increase in the
coreset size. Also, by applying similar techniques, we can obtain
a Streaming algorithm for the $k$-center problem
without outliers which uses  $\BO{k(1/\epsilon)^D}$ space
and features the same $(2+\epsilon)$-approximation as
\cite{McCutchen2008}. In
Section~\ref{sec:experiments} we compare the two algorithms experimentally. \\[0.1cm]

{\bf A 2-pass Streaming algorithm oblivious to $D$.}  As explained
before, thanks to its incremental nature, the MapReduce coreset
construction does not require explicit knowledge of the doubling
dimension $D$ of the metric space.  However, this is not the case for
the 1-pass Streaming algorithm described above, which requires the
apriori knowledge of $D$ to determine the proper value of $\tau$.
While in practice one can set $\tau$ to exercise suitable tradeoffs
between running time, working memory space and approximation quality,
it is of theoretical interest to observe that a simple-two pass
algorithm oblivious to $D$ with roughly the same bounds on the size of
the working memory can be obtained by ``simulating'' the 2-round
MapReduce algorithm for $\ell=1$.  

In the first pass, we run the doubling algorithm of
\cite{CharikarCFM04} for the $(k+z)$-center problem, thus obtaining a
radius value $\hat{r} \leq 8 r^*_{k+z}\leq 8r^*_{k,z}$.  Using
$\hat{r}$ as an estimate for $r^*_{k,z}$, in the second pass we
determine a maximal weighted coreset $T$ of points whose mutual
distances are greater than $(\epsilon/48)\hat{r}$. During the pass,
each point $s \in S-T$ is virtually assigned to a proxy in $T$ at
distance at most $(\epsilon/48)\hat{r}$, and for every $x \in T$ a
weight is computed as the number of points for which $x$ is
proxy. Finally, our weighted variant of the algorithm of
\cite{Charikar2001} is run on $T$.  It is easy to see that $|T| \leq
(k+z)(96/\epsilon)^D$ and that each point of $S$ is at distance at
most $\epsilon/6$ from its proxy. This immediately implies this
two-pass strategy returns a $(3+\epsilon)$-approximate solution to the
$k$-center problem with $z$ outliers with the same working memory
bounds as those stated in Theorem~\ref{thm:streaming-outliers} and
Corollary~\ref{corol:streaming-outliers}.

\section{Experiments} \label{sec:experiments}

In order to demonstrate the practical appeal of our approach, we
designed a suite of experiments with the following objectives:
(a) to assess the impact of coreset size 
on solution quality in our MapReduce and Streaming 
algorithms and to compare them to the state-of-the-art algorithms
for $k$-center with and without outliers (Subsections~\ref{subsec:kcenter} and \ref{subsec:outliers}, respectively);
(b) to assess the scalability of our MapReduce algorithms
(Subsection~\ref{subsec:scalability}); and (c)
to show that the MapReduce algorithm for  $k$-center without outliers
yields a much faster
sequential algorithm for the problem
(Subsection~\ref{subsec:run-sequential}).

\vspace{.3em}
\noindent\textbf{Experimental setting.}  The experiments were run on a
cluster of 16 machines, each equipped with a 18GB RAM and a 4-core
Intel I7 processor, connected by a 10GBit Ethernet network, using
Spark \cite{Zaharia2010} for implementing the MapReduce algorithms,
and a sequential simulation for the Streaming setting. We exercised our
algorithms on two low-dimensional real-world datasets used in
\cite{MalkomesKCWM15}, to facilitate the comparison with that work,
and on a higher-dimensional dataset as a stress test for our
dimension-sensitive strategies. The first dataset, 
\dataset{Higgs}~\cite{HiggsData}, 
contains 11 million points used to train learning algorithms for
high-energy Physics experiments.  The second dataset,
\dataset{Power}~\cite{PowerData}, 
contains 2,075,259 points which are measurements of electric
power consumptions in a house over four years.  The \dataset{Higgs}
dataset features 28 attributes, where 7 of them are a function of the
other 21. In~\cite{MalkomesKCWM15} only the 7 derived attributes were used: we do the same for the sake of comparison.
The \dataset{Power} dataset has 7 numeric attributes (we ignore the two non numeric features).
The third higher-dimensional dataset was
obtained from a dump of the English Wikipedia (dated December 2017) using the word2vec~\cite{Mikolov13} model with 50 dimensions.
This dataset, which we call \dataset{Wiki}, comprises 5,512,693 vectors.
To test the
scalability of our algorithms, we also generated artificially-inflated
instances of the \dataset{Higgs}, \dataset{Power}, and \dataset{Wiki} datasets (see details in
Subsection~\ref{subsec:scalability}). For all datasets we used the
Euclidean distance.  All numerical figures 
have been obtained as averages over at least 10 runs
and are reported in the graphs 
together with 95\% confidence intervals. The solution quality
is expressed in terms of the \emph{approximation 
ratio}, estimated empirically as the ratio between the radius of the returned clustering 
and the best radius ever found across \emph{all} experiments 
with the same dataset and parameter
configuration. (Note that the hardness of the problems
makes computing the actual optimal solution unfeasible.) The
source code of our algorithms is publicly available at
\url{https://github.com/Cecca/coreset-clustering}.

\subsection{$k$-center} \label{subsec:kcenter}
\begin{figure}[t]
    \includegraphics[width=\columnwidth]{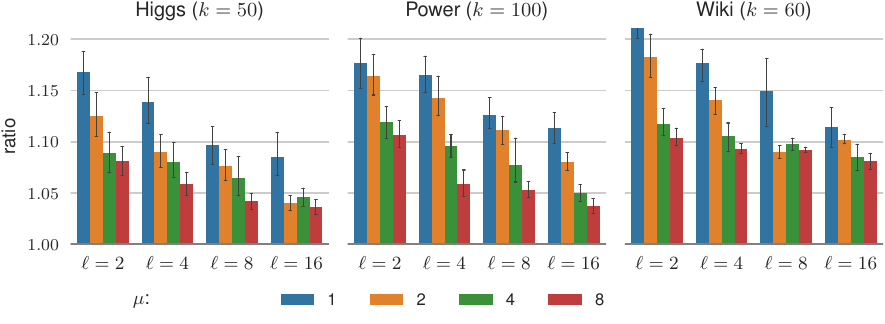}
    \caption{\sloppy Approximation
    ratio attained by the MapReduce algorithm for $k$-center
    using  
    coresets of size $\mu k$,  with $\mu = 1,2,4,8$, and 
    parallelism $\ell=2,4,8,16$.}
    \label{fig:k-center}
\end{figure}
We first evaluated the MapReduce algorithm for the $k$-center problem,
presented in
Subsection~\ref{sec-noout},
aiming at assessing the impact of the coreset size on the 
quality of the returned solution.  For
simplicity, rather than varying the precision parameter $\epsilon$, we varied 
the size of the coreset $T_i$ extracted from each partition $S_i$, 
setting it to the same value $\tau = \mu k$ for all $i$,
with $\mu=1,2,4,8$. Note that for
$\mu=1$ the algorithm corresponds to the one in \cite{MalkomesKCWM15}.
We fixed $k=50$ for the
\dataset{Higgs} dataset, $k=100$ for the \dataset{Power}
dataset, and $k=60$ for the \dataset{Wiki} dataset. These values of $k$, determined through a
number of experiments (omitted for brevity) have been chosen as
reasonable values marking the beginning of a plateau
in the radius of the clustering induced by the returned centers.
The plot in Figure~\ref{fig:k-center} reports the approximation
ratio attained by the algorithm for different coreset sizes
and degrees of parallelism.
As implied by the theory, the solution quality improves noticeably as
the size of the coreset (regulated by $\mu$) increases. Moreover, the
experiments show that, with respect to the algorithm by
\cite{MalkomesKCWM15} (blue bar
in the plot), even a moderate increase in the coreset size yields a
sensibly better solution. 
This behavior is observed also on the \dataset{Wiki} dataset, which, given its
high dimensionality, is a difficult input for our algorithm. 
In these experiments, the running times, not
reported for brevity, exhibited essentially a linear behavior in
$\tau$, for fixed parallelism, but remained tolerable (under one minute) even for
$\tau=8k$ and parallelism $\ell=2$.  Considering also the scalability
of the algorithm, which will be assessed in
Subsection~\ref{subsec:scalability}, we can conclude that using larger
coresets can yield better solution quality at a tolerable performance
penalty. From the figure, we finally 
observe that increasing the parallelism $\ell$ also leads to
better solutions, which is due to the fact that the size $\ell\cdot\tau$ of the
aggregated coreset $T$ on which \gmm is run in the second round, increases.  

\newcommand{\baseline}{\textsc{BaseStream}\xspace}
\newcommand{\ours}{\textsc{CoresetStream}\xspace}

\begin{figure}
  \centering
  \includegraphics[width=\columnwidth]{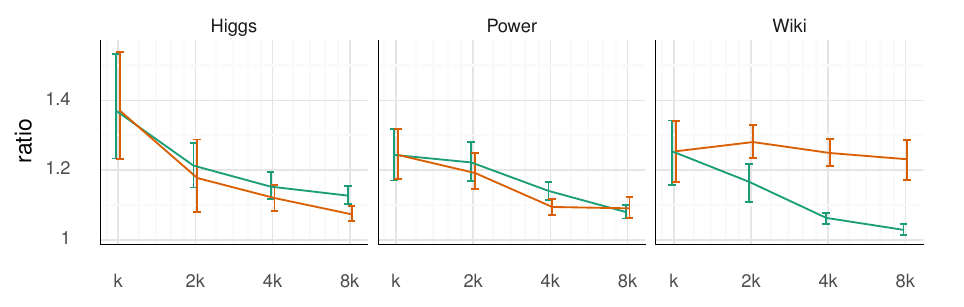}
  \includegraphics[width=\columnwidth]{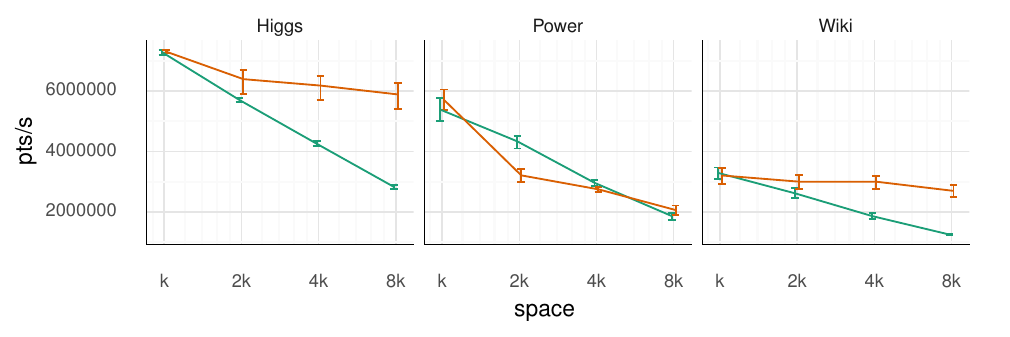}
  \caption{Approximation ratio (top) and throughput (bottom) versus space for 
the \ours (in orange) and \baseline (in green) $k$-center streaming algorithms. \ours uses space $\mu \cdot k$, 
with $\mu = 1,2,4,8,16$, \baseline
requires space $m\cdot k$, with $m=1,2,4,8,16$ ($\mu$ and $m$ increase from left to right in each plot).
}
\label{fig:no-outliers-streaming}
\end{figure}

For what concerns the Streaming setting, as observed in 
Section~\ref{sec:streaming}, our coreset approach would yield an
algorithm matching the approximation quality of the state-of-the-art
$(2+\epsilon)$-approximation algorithm by \cite{McCutchen2008}.
Nonetheless, we performed a number of experiments to compare the
practical performance of the two algorithms. The results,
reported in Figure~\ref{fig:no-outliers-streaming}, show that 
the algorithm by \cite{McCutchen2008} (dubbed {\sc BaseStream})
makes slightly better use of 
the available space, although our algorithm (dubbed {\sc CoresetStream}) often exhibits higher throughput
while yielding similar approximation quality.

\subsection{$k$-center with outliers} \label{subsec:outliers}
To evaluate our algorithms for the $k$-center problem with $z$
outliers, we artificially injected outliers into the datasets as
follows.  For each dataset, we first determined radius $r_{\mbox{\tiny MEB}}$
and center $c_{\mbox{\tiny MEB}}$ of its Minimum Enclosing Ball (MEB).  Then,
we added $z=200$ points at distance $100\cdot r_{\mbox{\tiny MEB}}$ from the
$c_{\mbox{\tiny MEB}}$ in random directions.  By doing so, each added point is
at distance $\ge 99\cdot r_{\mbox{\tiny MEB}}$ from any point in the
dataset. Furthermore, we verified that the minimum distance between
any two added points is $\ge 10\cdot r_{\mbox{\tiny MEB}}$, making these points
true outliers.

\begin{figure}
\centering
\includegraphics[width=\columnwidth]{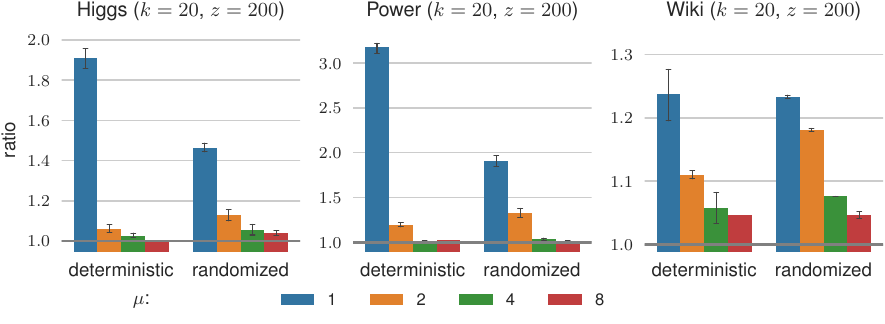}
\includegraphics[width=\columnwidth]{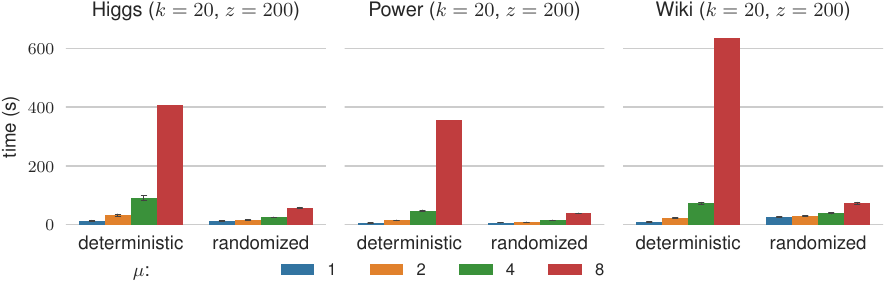}
\caption{Approximation ratio (top) and running time (bottom) attained by the deterministic and randomized MapReduce
algorithms for the $k$-center with $z$ outliers problem, 
using coresets of size
$\mu(k + z)$ and $\mu(k + 6\cdot z/\ell)$, respectively, with $\mu = 1,2,4,8$,
and fixed parallelism $\ell=16$.}
\label{fig:outliers-quality}
\end{figure}

A first set of experiments was run to compare the deterministic and
randomized versions of our algorithm presented in
Subsection~\ref{sec-out} against each other and against the algorithm
in \cite{MalkomesKCWM15}.  We set $k=20$ and $z=200$ for both datasets
and fixed the parallelism to $\ell=16$.  Also, we partitioned the data
adversarially, placing all outliers in the same partition so to better
test the benefits of randomization.  As before, rather than regulating
the size of each coreset $T_i$ through the precision parameter, we
fixed it equal to $\tau$ for each $i$, setting $\tau = \mu(k+z)$ for
the deterministic algorithm, and $\tau=\mu(k+ 6\cdot z/\ell)$ for the
randomized one, with $\mu=1,2,4,8$.  Again, the deterministic
algorithm with $\mu=1$ coincides with the algorithm by
\cite{MalkomesKCWM15}.  Based on Lemma~\ref{lem:occupancy}, the term
$6\cdot z/\ell$ in the value of $\tau$ for the randomized algorithm
is meant to upper bound the number of outliers included in each
partition (ignoring the logarithmic factor which is needed 
to ensure high probability only
when $z \simeq \ell$).

Figure~\ref{fig:outliers-quality} reports the results of these
experiments.  As before, we note that the quality of the solution
improves noticeably with the coreset size (regulated by $\mu$) and
even a moderate increase in the coreset size yields a significant
improvement with respect to the baseline of \cite{MalkomesKCWM15},
represented by the blue column ($\mu=1$, deterministic). In
particular, when $\mu=1$ the coreset extracted from the partition
containing all outliers is forced to include the outliers, hence few
other centers can be selected to account for the non-outlier points in
the partition, which are thus underrepresented. In this case, the
randomized algorithm, where the number of outliers per partition is
smaller and slightly overestimated by the constant 6, attains a better
solution quality.  As the coreset size increases, there is a sharper
improvement of the quality of the solution found by the deterministic
algorithm, since there are now enough centers to well represent the
non-outlier points, even in the partition containing all outliers,
while in the randomized algorithm, the effect of the coreset size on
the quality of the solution is much smoother.  Nevertheless, for $\mu
> 1$, the randomized algorithm finds solutions of comparable quality
to the ones found by the deterministic algorithm, using much smaller
coresets.  For what concerns the running time, the bottom plots of
Figure~\ref{fig:outliers-quality} clearly show that the reduction in
the coreset size featured by the randomized algorithm yields high
gains in performance, providing evidence that this algorithm can
attain much better solutions than \cite{MalkomesKCWM15} with a
comparable running time.


\renewcommand{\baseline}{\textsc{BaseOutliers}\xspace}
\renewcommand{\ours}{\textsc{CoresetOutliers}\xspace}

\begin{figure}
  \centering
  \includegraphics[width=\columnwidth]{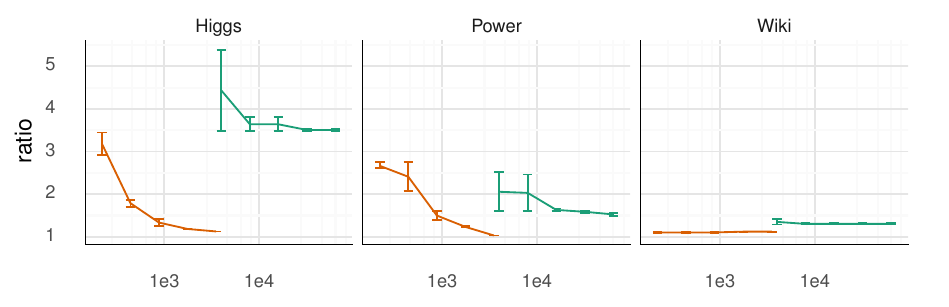}
  \includegraphics[width=\columnwidth]{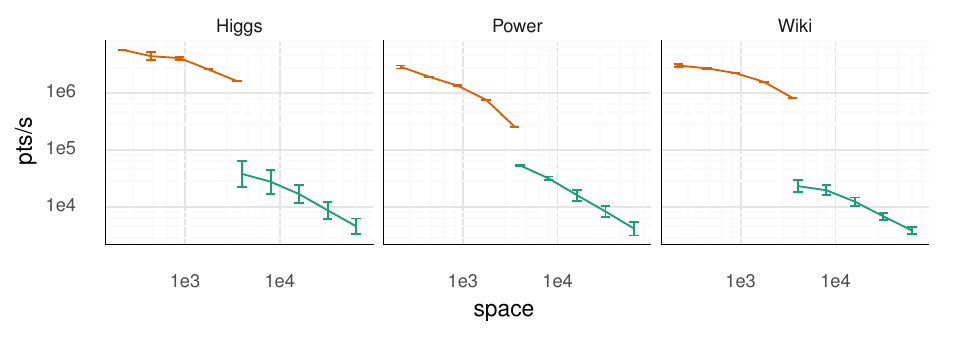}
  \caption{Approximation ratio (top) and throughput (bottom) versus space for 
    \ours (in orange) and \baseline (in green). \ours uses space $\mu(k + z)$, 
with $\mu = 1,2,4,8,16$, \baseline
requires space $m(k\cdot z)$, with $m=1,2,4,8,16$ ($\mu$ and $m$ increase from left to right in each plot).
Space and throughput are in logarithmic scale.
}
  \label{fig:outliers-streaming}
\end{figure}

In a second set of experiments, we studied the impact of the coreset
size on the quality of the solution computed by the Streaming
algorithm presented in Section~\ref{sec:streaming} (dubbed \ours)
and compared its performance with the state-of-the-art algorithm
of~\cite{McCutchen2008} (dubbed \baseline) which essentially runs a
number $m$ of parallel instances of a $(k\cdot z)$-space Streaming
algorithm, where $m$ depends on the desired approximation target.  We
used the same datasets and the same input parameters ($k=20$ and
$z=200$) as in the previous experiment. The points are shuffled before
being streamed to the algorithms.
Since the two algorithms feature different parameters, we compare
their performance as a function of the amount of space used, which is
$\mu(k+z)$ (i.e., the coreset size) for \ours, and $m(k\cdot z)$ for
\baseline.
The results are reported in Figure~\ref{fig:outliers-streaming}. 
We observe that for \dataset{Higgs} and \dataset{Power} \ours yields
better approximation ratios than \baseline using considerably less
space, which is coherent with the better theoretical quality featured
by the former algorithm.  For both algorithms, using more resources
(i.e., larger values of $\mu$ and $m$, respectively) leads to better
quality solutions, with \ours approaching the best quality ever
attained (approximation ratio almost 1). As for \dataset{Wiki}, we note
that both algorithms already yield very good solutions 
with minimum
space, which implies that for this dataset larger space does not
provide significant quality improvements. This is probably an effect of
the high dimensionality of the dataset.
To assess efficiency, we considered \emph{throughput}, i.e., the
number of points processed per second by the algorithm ignoring the
cost of streaming data from memory.  As expected, for both \ours and
\baseline throughput is inversely proportional to the space used.
However, by comparing the top and bottom graphs for each dataset,
it can be immediately seen that for a fixed approximation ratio, \ours
uses less space and exhibits a throughput substantially higher (always
more than 1 order of magnitude).  Thanks to its high throughout, even
for large values of $\mu$, \ours is able to keep up with real-world
streaming pipelines (e.g., in 2013 Twitter peaked at 143,199
tweets/s~\cite{TwitterThroughput}).

\subsection{Scalability of the MapReduce algorithms} \label{subsec:scalability}

For brevity, we focus on the randomized MapReduce algorithm for the
$k$-center problem with $z$ outliers, since the results for the other
cases are similar. A first set of experiments was run to assess the
scalability with respect to the input size.  To this end, we generated
synthetic instances of the \dataset{Higgs}, \dataset{Power}, and
\dataset{Wiki} datasets, $h$ times larger than the original datasets,
with $h =25, 50$ and 100. We used the following generation process.
Starting with the original dataset, a random point is sampled, and
each of its coordinates is modified through the addition of a Gaussian
noise term with mean 0 and standard deviation which is $10\%$ of the
difference between the maximum and the minimum value of that
coordinate across the original dataset. This perturbed point is then
added to the synthetic dataset until the desired size is reached.  The
rationale behind this construction is to build a (much larger)
synthetic dataset with the same clustered structure as the original
one, similarly to the SMOTE technique used in machine learning to
combat class imbalance \cite{ChawlaBHK02}. Also, outliers have been
added to each generated instance, as detailed in the previous
subsection.  On each instance of the datasets we ran the randomized
MapReduce algorithm with $k=20$, $z=200$, using maximum parallelism
($\ell=16$) and setting the size of each coreset $T_i$ to $8*(k +
6\cdot z/\ell)$.  Figure~\ref{fig:scalability-n} plots the running
times (averages of 10 runs) and shows that the algorithm scales
linearly with the input size.
\begin{figure}
  \centering
\includegraphics[width=\columnwidth]{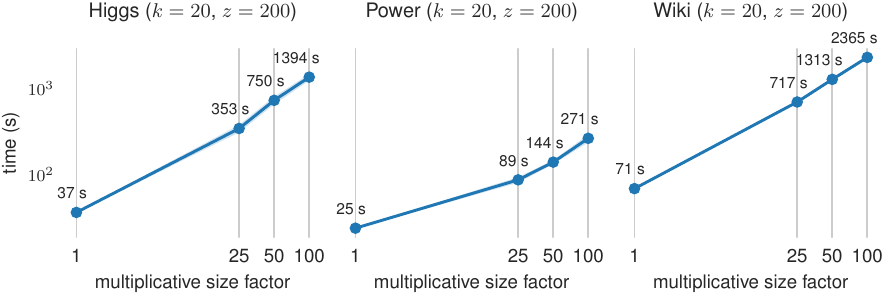}
\caption{Scalability with respect to input size
of the randomized MapReduce algorithm
for the $k$-center problem with $z$ outliers,
using coresets of size $8\cdot(k + 6\cdot z/\ell)$ and parallelism 16.
Both axes are in logarithmic scale.}
  \label{fig:scalability-n}
\end{figure}

We ran a second set of experiments to assess the scalability of the algorithm
with respect to the number of processors.  For these experiments, we
used the original datasets with added outliers, setting 
$k=20$ and $z=200$, as before. In order to target the same
solution quality over all runs, we fixed the size of the union of the
coresets, from which \outliers{} extracts the final solution, equal to
$8(16k + 6z)$, which corresponds to the case $\mu=8$ and $\ell=16$ of
Figure~\ref{fig:outliers-quality}.  Then, we ran the algorithm varying
the parallelism $\ell$ between $1$ and $16$, setting, for each value
of $\ell$, the size of each $T_i$ to
$\tau_{\ell}=8(16k + 6z)/\ell$,
so to obtain the desired size for the union.
Figure~\ref{fig:scalability-p} plots the running times 
distinguishing between the time required by the coresets construction
(orange area)
and the time required by \outliers{} (blue area). 
While the latter time is clearly constant, 
coreset construction time, which 
dominates the running time for small $\ell$, scales superlinearly with
the number of processors. In fact, doubling the parallelism results in
about a 4-fold improvement of the running time up to 
8 processors, since
each processor performs work proportional to $\tau_{\ell} \cdot |S|/\ell$,
and  $\tau_{\ell}$ embodies an extra factor $\ell$
in the denominator. This effect is milder going from 8 to 16 processors because
of the overhead of initial random shuffle of the data.

\begin{figure}[t]
  \centering
  \includegraphics[width=\columnwidth]{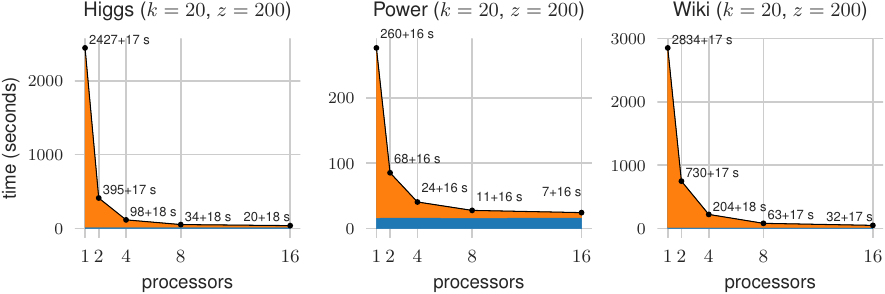}
  \caption{Scalability with respect to the number of processors
of the randomized MapReduce algorithm
for the $k$-center problem with $z$ outliers,
using coresets of size $8*(16k + 6z)$ and parallelism 
$\ell=1,2,4,8,16$. Each point is labeled with the sum of 
the time required to build the coreset
(orange area) and the time required to compute the final solution on the coreset
(blue area).}
  \label{fig:scalability-p}
\end{figure}

\subsection{Improved sequential performance}
\label{subsec:run-sequential}
As we discussed in Section~\ref{sec:MR}, for the $k$-center problem
with $z$ outliers we can improve on the superquadratic complexity of the state of the art algorithm in \cite{Charikar2001}, which we dub \textsc{CharikarEtAl} in the following, by running our deterministic MapReduce algorithm sequentially,
at the expense of a slightly worse approximation
guarantee. 
(In fact, the \textsc{CharikarEtAl} algorithm amounts to $\BO{\log|S|}$ executions of our \outliers with $\hat{\epsilon}=0$ and unit weights on the entire input $S$.)
To quantify the achievable gains, we took a sample of 10000
points from each dataset (so to keep \textsc{CharikarEtAl}'s running time within
feasible bounds).  As before, we injected 200 outliers, using the same
procedure outlined above, and set $k=20$ and $z=200$. We ran our
MapReduce algorithm with $\ell=1$ (indeed, for $\ell=1$, the algorithm
is sequential)
and $\mu=1,2,4,8$. 
Figure~\ref{fig:sequential-comparison} reports, for the three datasets,
the running times 
(top plots) and the radii of 
the returned clusterings (bottom plot) for
\textsc{CharikarEtAl}  
and our algorithm for varying $\mu$.
Measures are averages over 10 runs, with the input dataset shuffled before each run.
Note that the case
$\mu=1$ corresponds to the
algorithm in \cite{MalkomesKCWM15}, therefore
we label it as \textsc{MalkomesEtAl} 
\begin{figure}[t]
  \centering
  \includegraphics[width=\columnwidth]{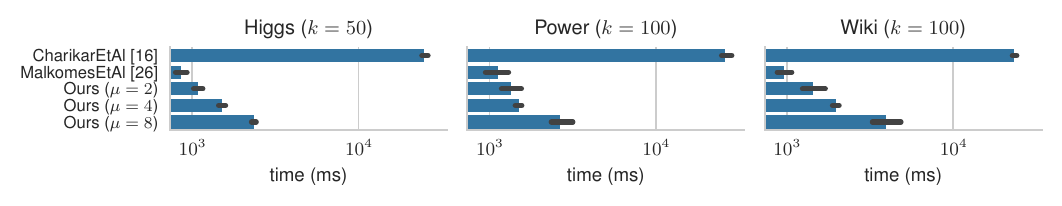}
  \includegraphics[width=\columnwidth]{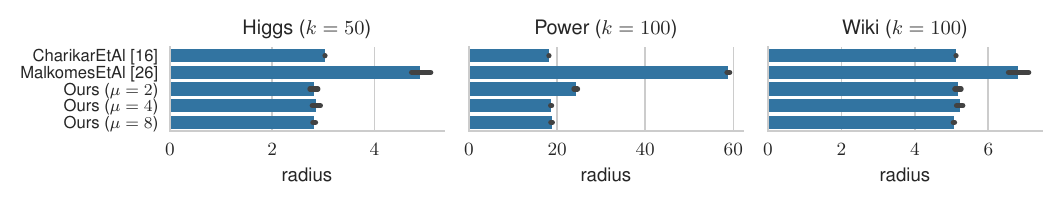}
  \caption{\label{fig:sequential-comparison} Running time (top, in logarithmic scale) and radius (bottom) of different sequential algorithms on a sample of 10 thousands points of \dataset{Higgs}, \dataset{Power}, and \dataset{Wiki}.}
  \vspace{-1em}
\end{figure}
From the figure it is clear that 
building a coreset before running \outliers 
is highly beneficial for the running time, which improves by
one order of magnitude. However, the solution quality for 
\textsc{MalkomesEtAl} (i.e., $\mu=1$)
 is much worse than
the one featured by \textsc{CharikarEtAl}.  In contrast, the bars for $\mu > 1$
show that a substantial performance improvement over the one of
\textsc{CharikarEtAl} can be attained, 
while keeping the approximation quality
essentially unchanged.  Observe that, in some cases, our algorithm
returns better radii than \textsc{CharikarEtAl}, even if from the theory one would expect a slightly
worse behavior. This is probably due to the fact that while
\textsc{CharikarEtAl} is essentially insensitive to shufflings of the data, our
coreset construction, based on \gmm, introduces an element of
arbitrariness with the choice of the initial center, which may result
in different coresets for different shuffles, 
potentially leading to a better average
solution quality.

\section{Conclusions}\label{sec:conclusions}

We presented MapReduce and Streaming algorithms for the $k$-center
problem (with and without outliers) based on flexible coreset
constructions. These constructions yield a spectrum of space-accuracy
tradeoffs regulated by the doubling dimension $D$ of the underlying
space, and afford approximation guarantees arbitrarily
close to those of the best sequential strategies, using moderate space
in the case of small $D$. The theoretical analysis of the algorithms is
complemented by experimental evidence of their practicality.

\ifextended
Coresets provide an effective way of processing large amounts of data
by building a succinct summary of the input which can then be processed
with the sequential algorithm of choice. In particular, we showed how
to leverage coresets to build MapReduce and Streaming algorithms for
the $k$-center problem with and without outliers.  Building on
state-of-the art approaches for these problems, we provide flexible
coreset constructions which yield a spectrum of space-accuracy
tradeoffs which allow to obtain approximation guarantees that can be
made arbitrarily close to those obtainable with the best sequential
strategies at the expense of an increase of the memory requirements,
regulated by the dimensionality of the underlying metric space. The
theoretical findings are complemented by experimental evidence of the
practicality of the proposed algorithms.
\fi

Future avenues of research include further improvements of the local
memory requirements of the MapReduce algorithms, the development of a
1-pass Streaming algorithm oblivious to the doubling dimension $D$ of
the metric space, and the extension of our approach to other
(center-based) clustering problems.

\bibliographystyle{abbrv}
\bibliography{references}

\end{document}